\documentclass[a4paper,11pt,reqno]{amsart}
\pdfoutput=1

\usepackage{natbib,graphicx,ifthen,multirow,hyperref,doi,amsaddr,thmtools}
\usepackage[left=1in,right=1in,top=1.4in,bottom=1.2in]{geometry}

\hypersetup{
	pdftitle={Hyperbolic normal stochastic volatility model},
	pdfauthor={Jaehyuk Choi, Chenru Liu, Byoung Ki Seo},
	pdfkeywords={stochastic volatility, SABR model, Bougerol's identity, Johnson's SU distribution},
	colorlinks=true,
	linkcolor=red,
	citecolor=blue,
	urlcolor=blue,
	bookmarksnumbered=true,
	pdfstartview=
}

\newtheorem{prop}{Proposition}
\newtheorem{coro}{Corollary}

\newcommand{\nondim}[1]{\tilde{#1}}
\newcommand{\rhoc}{\rho_\ast}
\newcommand{\sigman}{\tilde{\sigma}}

\newcommand{\distequal}{\,{\buildrel d \over =}\,}

\newcommand{\Zdrift}[1][\mu]{\ifthenelse{\equal{#1}{0}}{Z}{Z^{[#1]}}}
\newcommand{\ExpBM}[1][\mu]{\ifthenelse{\equal{#1}{0}}{A}{A^{[#1]}}}
\newcommand{\asinh}{\mathrm{asinh}}
\newcommand{\acosh}{\mathrm{acosh}}
\newcommand{\atanh}{\mathrm{atanh}}

\newcommand{\fwd}{\bar{F}_T}
\newcommand{\ExpS}{w}
\newcommand{\drift}{\lambda}
\newcommand{\Hplane}[1]{\mathbb{H}_{#1}}
\newcommand{\qtext}[2][\quad]{#1\text{#2}#1}

\title{Hyperbolic normal stochastic volatility model}

\author[Choi]{Jaehyuk Choi}
\address{Peking University HSBC Business School}
\email{\href{mailto:jaehyuk@phbs.pku.edu.cn}{jaehyuk@phbs.pku.edu.cn}}

\author[Liu]{Chenru Liu}
\address{Department of Management Science and Engineering, Stanford University}
\email{\href{liucr@stanford.edu}{liucr@stanford.edu}}

\author[Seo]{Byoung Ki Seo}
\address{Ulsan National Institute of Science and Technology}
\email{\href{mailto:bkseo@unist.ac.kr}{bkseo@unist.ac.kr}}

\date{September 7, 2018}

\begin{document}

\begin{abstract}
For option pricing models and heavy-tailed distributions, this study proposes a continuous-time stochastic volatility model based on an arithmetic Brownian motion: a one-parameter extension of the normal stochastic alpha-beta-rho (SABR) model. Using two generalized Bougerol's identities in the literature, the study shows that our model has a closed-form Monte-Carlo simulation scheme and that the transition probability for one special case follows Johnson's $S_U$ distribution---a popular heavy-tailed distribution originally proposed without stochastic process. It is argued that the $S_U$ distribution serves as an analytically superior alternative to the normal SABR model because the two distributions are empirically similar.
\end{abstract}

\maketitle

\section{Introduction}
Stochastic volatility (SV) models have been proposed to overcome the failure of the Black-Scholes-Merton (BSM) model in explaining non-constant implied volatilities across strike prices on option markets, a phenomenon called volatility smile. Therefore, most previous studies (e.g., \citet{hullwhite1987sv,stein1991sv,heston1993closed}) discuss the SV models based on the geometric Brownian motion (BM) (hereafter, lognormal SV models).
On the other hand, the studies on the SV models that are based on the arithmetic BM (hereafter, normal SV models) are scarce. This study aims to fill this gap by proposing and analyzing a class of normal SV models. Our motivation for choosing arithmetic BM as the \textit{backbone} of the SV model is twofold: 
an options pricing model alternative to the lognormal SV model and a skewed and heavy-tailed distribution generated by a continuous-time stochastic process.

\subsection{Options pricing model}
First, the study discusses the aspect of options pricing model. Although eclipsed by the success of the BSM model, the arithmetic BM is analyzed for the first time as an options pricing model by \citet{bachelier} (hereafter, normal model) and still provides more relevant dynamics than the geometric BM for some financial asset classes. 
Refer to \citet{brooks2017option} and \citet{schachermayer2008close} for recent surveys on the normal model. An important difference between them is that the volatility under the normal model (hereafter, normal volatility) measures the uncertainty in terms of the absolute change in the asset price as opposed to relative change.
One example of the applications of the normal model is its use for modeling the interest rate. The proportionality between the daily changes and level of interest rate---a key assumption of the BSM model---is empirically weak~\citep{levin2004rate}. 
Therefore, among fixed-income market traders, the normal model has long been a popular alternative to the BSM model for quoting and risk-managing the options on interest rate swap and Treasury bonds (and futures). For example, the Merrill Lynch option volatility index (MOVE)---the bond market's equivalent of the volatility index (VIX)---is calculated as the weighted average of the implied normal volatilities of the US Treasury bond options. It is also worth noting that the hedging ratio, delta, from the normal and BSM models can often be significantly different, even after the volatilities of the corresponding models are calibrated to the same option price observed on the market. Therefore, the normal model's delta provides a more efficient hedge when the fluctuation of the underlying asset price is more consistent in absolute term than in percentage term.
The use of the normal model for the interest market is further justified by the negative policy rates observed in several developed economies after the global financial crisis of 2008. Other than the interest rate, the normal model is often used for modeling the inflation rate~\citep{kenyon2008inflation} and spread option~\citep{poitras1998spread}.

Despite this background, it is difficult to find previous studies on the normal SV model. It is surprising, given that the lognormal SV models are often analyzed under the normal diffusion framework with the log price transformation; it implies that any existing results on the lognormal SV models can be effortlessly applied to the corresponding normal SV models. To the best of our knowledge, the only previous study on the normal SV model is in the context of the stochastic alpha-beta-rho (SABR) model~\citep{hagan2002sabr}---an SV model popular among practitioners. In the SABR model, the price follows a constant elasticity of variance (CEV) backbone, while the volatility follows a geometric BM. Therefore, the SABR model provides a range of backbone choices, including the normal and lognormal backbones. The SABR model with normal backbone (hereafter, normal SABR) is an important motivation for this study. A detailed review on the SABR model is provided in section~\ref{sec:sabr}.

\subsection{Skewed and heavy-tailed distribution}
The second motivation of our study is that the normal SV models can serve as a means to generate distributions with skewness and heavy-tail, generalizing the normal distribution. Heavy-tailed distributions are ubiquitous and their importance cannot be emphasized enough. In this regard, the study of normal SV models has a much broader significance than that of the lognormal SV models. This is because the latter generalizes the lognormal distribution whose application is limited when compared to the normal distribution.

Several distribution families have been proposed in statistics to incorporate skewness and heavy tails into a normal distribution. Even if the focus is narrowed to the applications to finance, it can be found that numerous distributions have been adopted to describe the statistics of asset return: 
generalized lambda ~\citep{corlu2015modelling}, stable \citep{fama1965behavior}, skewed $t$ \citep{theodossiou1998}, Gaussian mixture \citep{kon1984models,behr2009alt}, generalized hyperbolic \citep{eberlein1995hyperbolic,behr2009alt}, Turkey's $g$- and $h$- \citep{badrinath1988,mills1995modelling}, and Johnson's $S_U$ (hereafter $S_U$) distribution \citep{shang2004modeling,gurrola2007cap,choi2008asymmetric}.

However, the above distributions are neither defined from or associated with stochastic differential equations (SDEs), not to mention the SV models in particular. Those distributions are defined by the probability density function (PDF) or by the transformations of other well-known random variables. This is because it is usually difficult for an SDE to yield an analytically tractable solution. There are only a few examples of continuous-time processes whose transition probabilities correspond to the following well-known probability distributions: the arithmetic BM to a normal distribution (by definition), geometric BM to a lognormal distribution, and CEV and CIR processes to non-central $\chi^2$ distributions.

\subsection{Contribution of this study}
The study proposes and analyzes a class of normal SV models, which includes the normal SABR model as a special case. Since the mathematics behind our model involves the BMs in hyperbolic geometry and the results are expressed by hyperbolic functions, the class is named \textit{hyperbolic normal SV} or \textit{NSVh}~\footnote{The class is named as an abbreviation in a manner similar to the way hyperbolic sine becomes $\sinh$} model. The important mathematical tool to analyze the NSVh model comes from the two generalizations~\citep{alili1997iden,alili1997} of Bougerol's identity~\citep{bougerol1983}. 

The first generalization leads us to a closed-form Monte-Carlo (MC) simulation scheme that no longer needs a time-discretized Euler scheme. The MC scheme requires merely one and a half (1.5) normal random numbers for a transition between time intervals of any length. Although limited to the normal SABR case, this study's scheme is far more efficient than the previous exact MC scheme of \citet{cai2017sabr}. Additionally, the original proof of the first generalization~\citep{alili1997} is simplified in this study.
The second generalization shows that a special case of the NSVh model--- different from the normal SABR model---gives rise to the $S_U$ distribution~\citep{johnson1949systems}, one of the popular heavy-tailed distributions. This allows the study to add to the literature one rare example of analytically tractable SDEs. The normal SV model provides a framework to understand the $S_U$ distribution in a better manner; the distribution can be parametrized more intuitively by using the NSVh parameters, and the popular use of the distribution is explained to some extent. 

Importantly, under the NSVh model framework, two unrelated subjects are brought together, that is, the normal SABR model and the $S_U$ distribution. It is argued and empirically shown that the two distributions are very close to each other when parameters are estimated from the same data set, and can thus be used interchangeably. Among the benefits of the interchangeable usage of the two is the superior analytic tractability of the $S_U$ distribution when \textit{recognized} as an options pricing model---various quantities of interest, such as the vanilla option price, density functions, skewness, ex-kurtosis, value-at-risk, and expected shortfall, have closed-form expressions that are not available in other SV models. To facilitate the interchangeability, a quick method of moments to convert the equivalent parameter sets between the two distributions is proposed.

This remainder of this paper is organized as follows. Section~\ref{sec:model} defines the NSVh model and reviews the SABR model and $S_U$ distribution. Section~\ref{sec:main} describes the main results. Section~\ref{sec:num} presents the numerical results with empirical data. Finally, Section~\ref{sec:conc} concludes the paper.

\section{Models and Preliminaries} \label{sec:model}

\subsection{NSVh Model} \label{sec:nsvh}
The NSVh model is introduced as
\begin{equation}
dF_t = \sigma_t \left(\rho\,d\Zdrift[\drift\alpha/2]_t + \rhoc\,dX_t\right) \quad \text{and}\quad
\frac{d\sigma_t}{\sigma_t} = \alpha\; d\Zdrift[\drift\alpha/2]_t,
\end{equation}
where $F_t$ and $\sigma_t$ are the processes for the price and volatility, respectively, $\alpha$ is the volatility of the volatility parameter, $\rho$ denotes the instantaneous correlation between $F_t$ and $\sigma_t$, and $\rhoc = \sqrt{1-\rho^2}$. The BMs $Z_t$ and $X_t$ are independent, and $\Zdrift_t = Z_t + \mu\, t$ denotes BM with drift $\mu$. 

The role of the model parameters $\rho, \alpha$, and $\drift$ is discussed. Similar to the lognormal SV models, correlation $\rho$ accounts for the asymmetry in the distribution, that is, skewness or volatility skew. The leverage effect---the negative correlation between the spot price and volatility seen in the equity market---is achieved by a negative $\rho$, although it is in the context of the normal volatility in the NSVh model. The parameter $\alpha$ accounts for the heavy tail, that is, excess kurtosis or volatility smile. It can be easily seen that the process converges to an arithmetic BM in the limit $\alpha\rightarrow 0$ regardless of $\rho$. Therefore, $\alpha$ affects both the skewness and heavy tail at the same time. 

The parameter $\drift$ is present in the drift of $Z_t$ for both $F_t$ and $\sigma_t$. With regards to the volatility process, $\drift$ controls the power of $\sigma_t$ that becomes a martingale as a geometric BM:
\begin{gather*}
\frac{d(\sigma_t)^{1-\drift}}{(\sigma_t)^{1-\drift}} = (1-\drift)\alpha \; dZ_t \quad \text{if} \quad \drift\neq 1,\\
d(\log \sigma_t) = \alpha\; dZ_t  \quad \text{if} \quad \drift = 1.
\end{gather*}
For example, $\drift=0$ yields the volatility $\sigma_t$, following a driftless geometric BM, as in the SABR model, and $\drift=-1$ yields the variance $\sigma_t^2$, following a driftless geometric BM as in the SV model of \citet{hullwhite1987sv}. With regards to the price process, however, the drift prevents $F_t$ from being a martingale except for $\drift=0$ or, less importantly, $\rho=0$, although the expectation is easily computed as $\bar{F}_t = F_0 + (\sigma_0 \rho/\alpha) \big( e^{\drift \alpha^2 t/2} - 1 \big)$.
Therefore, the resulting process may not be desirable as a price process. The NSVh model for $\drift\neq 0$ is understood as a probability distribution perturbed from the $\drift=0$ case, by applying the Radon-Nikodym derivative with respect to $\Zdrift[0]_t$. Essentially, the introduction of $\drift$ does not significantly diversify the shape of the distribution, and, therefore $\drift$ is not meant for parameter estimation. As we shall see, however, $\drift$ plays an important role in \textit{model selection}; it brings under one unified process the three subjects separately studied: the normal SABR model ($\drift=0$), Johnson's $S_U$ distribution ($\drift=1$), and BM on three-dimensional hyperbolic geometry ($\drift=-1$).

To provide a background for the main result in section~\ref{sec:main}, we simplify the SDEs into the canonical forms,
\begin{equation}
d\nondim{F}_s = \nondim{\sigma}_s (\rho\,d\Zdrift[\drift/2]_s + \rhoc\,dX_s) \qtext{and}
\frac{d\nondim{\sigma}_s}{\nondim{\sigma}_s} = d\Zdrift[\drift/2]_s \quad (\nondim{\sigma}_0 = 1),
\end{equation}
where the following changes of variables are used:
$$ s = \alpha^2 t, \quad \sigman_s = \frac{\sigma_t}{\sigma_0},
\qtext{and} \nondim{F}_s = \frac{\alpha}{\sigma_0} \big(F_t - \fwd\big).$$
Here, the new variable $s$ is the integrated variance of the log volatility, $\nondim{F}_s$ and $\nondim{\sigma}_s$ are the non-dimensionalized price and volatility processes, respectively, under $s$. The scaling of $\nondim{F}_s$ and $\nondim{\sigma}_s$ naturally follows from the time change of the BMs with the new variable $s$. The time $T$ is any fixed time of interest such as the time-to-expiry of the vanilla option. The price $F_t$ is shifted by $\fwd$ first to ensure that $E(\nondim{F}_S)=0$ at the corresponding time $S=\alpha^2 T$. Therefore, the canonical NSVh distribution is effectively parametrized by the three parameters---$(S,\; \rho,\; \drift)$. Additionally, the original distribution is recovered by $F_T = (\sigma_0/\alpha) \nondim{F}_{\alpha^2 T} + \fwd$ and $\sigma_T = \sigma_0\, \nondim{\sigma}_{\alpha^2 T}$. While the original and canonical representations are explicitly distinguished, the variable $S$ is often used in the original form as well for the sake of concise notation. 

The stochastic integrals of the canonical forms up to $s=S$ are, respectively, expressed as
\begin{equation} \label{eq:F_S}
\nondim{F}_S \; =\; \rho\,\big( e^{\Zdrift[(\drift-1)/2]_S} -  e^{\frac12\drift S} \big) + \rhoc\, X_{\ExpBM[(\drift-1)/2]_S} \qtext{and}
\nondim{\sigma}_S = \exp\left(\Zdrift[(\drift-1)/2]_S\right).
\end{equation}
It must be noted that the integral of $\nondim{\sigma}_s$, with respect to $X_s$, is further simplified to the BM time-changed with an exponential functional of the BM defined by
\begin{equation} \label{eq:ExpBM}
	\ExpBM_T = \int_{t=0}^T e^{2\Zdrift_t} dt \quad (\ExpBM[0]_T = \ExpBM[ 0]_T).
\end{equation}
This quantity has been the topic of extensive research; see \citet{matsuyor2005exp1, matsuyor2005exp2, yor2012exp} for a detailed review. While the functional is originally defined as the continuously averaged price under the BSM model in Asian options, it is used for the time-integral of the variance in the context of this study. Although $\ExpBM_T$ can be defined with \textit{any} standard BM, we implicitly assume that $\ExpBM_T$ is tied to a particular BM, $Z_t$, throughout this study. Essentially, $\ExpBM_T$ and $Z_T$ are closely intertwined, and the knowledge of their joint distribution of is the key to solve Equation~(\ref{eq:F_S}).

\subsection{SABR Model and Hyperbolic Geometry} \label{sec:sabr}
The SABR model is reviewed with a focus on the normal backbone along with the BM on hyperbolic geometry, which serves as a mathematical tool for the NSVh and SABR models. The SABR model~\citep{hagan2002sabr} is an SV model with the backbone of the CEV process:
\begin{equation}
\frac{dF_t}{F_t^\beta} = \sigma_t \, (\rho dZ_t + \rhoc dX_t) \qtext{and} \frac{d\sigma_t}{\sigma_t} = \alpha\, dZ_t,
\end{equation}
where $X_t$ and $Z_t$ are independent BMs. As mentioned earlier, the normal SABR model with $\beta=0$ is equivalent to the NSVh model with $\drift=0$. 

The SABR model has been widely used in the financial industry, for covering fixed income in particular, due to several merits: (i) arbitrary backbone choice, including normal ($\beta=0$) and lognormal ($\beta=1$) ones, (ii) availability of an approximate but fast vanilla options pricing method~\citep{hagan2002sabr}, and (iii) parsimonious and intuitive parameters. The comments on those merits are presented in order. Regarding the CEV backbone, the popularity of the SABR model provides another evidence that the lognormal backbone of the BSM model is not a one-fits-all solution. The normal SABR in this study allows the negative value of $F_t$ without any boundary condition at zero. It should not be confused with the continuous limit of $\beta\rightarrow 0^+$, which does not allow a negative value. In the original article, \citet{hagan2002sabr} derives an approximate formula for the implied BSM volatility, from which the option price can be quickly computed through the BSM formula. However, it is worth noting that the normal volatility is first obtained from the small-time perturbation of the normal diffusion even for $\beta\neq 0$ and subsequently it is converted to the BSM volatility by another approximation~\citep{hagan1999equiv}. Therefore, the option price computed from the normal volatility and the normal model formula has been considered more accurate because the second approximation can be avoided. Since the normal volatility is more appropriate for this study, the normal volatility approximation for $\beta=0$ is presented for reference and later use:
\begin{equation} \label{eq:hagan}
\begin{gathered}
\sigma_{N}(\sigma_0, \alpha, \rho, K) = \sigma_0 \left(\frac{\zeta}{\chi}\right)\left(1+\frac{2-3\rho^2}{24}\alpha^2 T\right)\\
\text{where}\quad \zeta = \frac{\alpha}{\sigma_0}(F_0-K)\quad\text{and}\quad\chi = \log \left( \frac{\sqrt{1-2\rho\zeta+\zeta^2}-\rho+\zeta}{1-\rho} \right),
\end{gathered}
\end{equation}
where $K$ is the strike price, and $T$ denotes the time-to-expiry. The volatility approximation is an asymptotic expansion that is valid when $\alpha^2 T$ is small; therefore, the accuracy of the approximation noticeably deteriorates with an increase in $\alpha^2 T$. Despite the shortcoming, the inaccurate approximation does not cause problems in pricing vanilla options because the model parameters $\sigma_0, \alpha, \rho$ and the pre-determined $\beta$ are to be calibrated to the option prices observed from the market. In this regard, the implied volatility formula rather serves as an interpolation method for the volatility smile. Inaccurate approximation starts causing issues only when the usage of the model goes beyond vanilla options pricing. Two such cases are as follows:(i) claims with an exotic payout (e.g., quadratic) that require the knowledge of PDF, and (ii) path-dependent claims, which must resort to an MC simulation. In the first case, the PDF implied from \citet{hagan2002sabr}'s formula often results in negative density at out-of-the-money strikes, thereby allowing arbitrage. In the second case, the vanilla option price from the formula is not consistent with that from the MC simulation with the same parameters. Therefore, the parameter calibration for MC scheme should be performed with extra care. Hence, the research on the accurate option analytics and efficient MC simulation methods come after the SABR model establishes its popularity among practitioners. 

The study reviews prior works on the SABR model. Concerning vanilla options pricing, there have been various improvements to \citet{hagan2002sabr}'s result. A few examples of such studies are
\citet{obloj2007fine,jordan2011vol,balland2013sabr,lorig2015lsv}. However, they remain as approximations. 
The exact pricing is known only for the following three special cases: (i) zero correlation ($\rho=0$), (ii) lognormal SABR ($\beta=1$), and (iii) normal SABR ($\beta=0$).
For the rest of the parameter ranges, no analytic solution is reported. Hence, the finite difference method~\citep{park2014sabr,floc2014fd} is considered the most practical approach.
Concerning the zero-correlation case, the price process can be transformed to the CEV process time-changed with $\ExpBM[-1/2]_T$ in a manner similar to that of Equation~(\ref{eq:F_S}). Thus, the option price is expressed by a multi-dimensional integral representation of $\ExpBM[-1/2]_T$ over the CEV option prices~\citep{schroder1989comp}. Refer to \citet{antonov2013sabr} for the most simplified expression based on the heat kernel on the two-dimensional hyperbolic geometry~\citep{mckean1970upper}, which we introduce below. The solution for the lognormal SABR is expressed in terms of the Gaussian hypergeometric series~\citep{lewis2000sv}.

The options pricing under the normal SABR model depends on the mathematical tools developed for the BMs on hyperbolic geometry, represented as Poincar\'e half-plane. Since the NSVh model also benefits from the same tools, we briefly introduce these tools. The $n$-dimensional Poincar\'e half-plane is denoted by $\Hplane{n}$. Table~\ref{tab:hyp_geo} provides a \textit{quick reference} for the properties of $\Hplane{2}$ and $\Hplane{3}$. The standard BM in a geometry is defined to be the stochastic process whose infinitesimal generator is given by the Laplace-Beltrami operator $\Delta$ of the geometry. The heat kernel $p(t,D)$ is the fundamental solution of the diffusion equation $(\partial_t - \frac12 \Delta) p(t,D) = 0$, and hence the transition probability of the standard BM. Table \ref{tab:hyp_geo} shows the heat kernels on $\Hplane{2}$~\citep{mckean1970upper}, often referred to as the McKean kernel, and $\Hplane{3}$ \citep{debiard1976theoremes}. The analytical expressions for the heat kernels are also known for $\Hplane{n}$, in general; see \citet{grigor1998heat} for derivation. 

\begin{table}
\caption{Properties of the $n$-Dimensional Hyperbolic Geometry Represented by $n$-Dimensional Poincar\'e Half-Plane $\Hplane{n}$ for $n=2$ and $3$. Symbols $\partial_x$ and $\partial_x^2$ are the shortened notations for partial derivative operators $\frac{\partial}{\partial x}$ and $\frac{\partial^2}{\partial x^2}$, respectively.}
\label{tab:hyp_geo}
\begin{center}
\begin{tabular}{|c||c|c|} \hline
	Dimension & $\Hplane{2} = \{(x,z):z>0\}$ & $\Hplane{3} = \{(x,y,z):z>0\}$ \\ \hline\hline
	Metric $(ds)^2$ & $(dx^2+dz^2)/z^2$ & $(dx^2+dy^2+dz^2)/z^2$ \\ \hline
	Volume element $dV$ & $dx\,dz/z^2$ & $dx\,dy\,dz/z^3$ \\ \hline
	Geodesic distance $D$& 
	\multirow{2}{*}{$\acosh \left(\frac{(x'-x)^2+z^2+z'^2}{2z z'} \right)$} &
	\multirow{2}{*}{$\acosh \left(\frac{(x'-x)^2+(y'-y)^2+z'^2+z^2}{2z z'} \right)$} \\
	$(x,\cdot,z)$ to $(x',\cdot,z')$ &  & \\ \hline
	Laplace-Beltrami &
	\multirow{2}{*}{$z^2\,(\partial_x^2 + \partial_z^2)$} &
	\multirow{2}{*}{$z^2\,(\partial_x^2 + \partial_y^2 + \partial_z^2) - z\, \partial_z$} \\
	operator $\Delta_{\Hplane{n}}$ & & \\ \hline
	\multirow{2}{*}{Standard BM} & $dx_t = z_t dX_t,$ & $dx_t = z_t dX_t, \quad dy_t = z_t dY_t,$ \\
	& $dz_t/z_t = dZ_t$ & $dz_t/z_t = dZ_t - dt/2$ \\ \hline 
	Heat kernel $p_n(t,D)$ &
	\multirow{3}{*}{$\displaystyle \frac{\sqrt2 e^{-t/8}}{(2\pi t)^{3/2}} \int_D^\infty\!\!  ds \frac{s e^{-s^2/2t}}{\sqrt{\cosh s - \cosh D}}$} &
	\multirow{3}{*}{$\displaystyle  \frac1{(2\pi t)^{3/2}} \frac{D}{\sinh D} e^{-(t^2+D^2)/2t}$} \\
	for $n=2$ or $3$ & & \\
	$(\partial_t - \frac12 \Delta_{\Hplane{n}}) p_n = 0$& & \\ \hline
\end{tabular}
\end{center}
\end{table}

The standard BM on $\Hplane{2}$ is equivalent to the normal SABR with $\rho=0$ in canonical form, where the $x$-axis is for the price process and the $z$-axis for the volatility process. Naturally, the $\Hplane{2}$ heat kernel has been used for the analysis of the normal SABR. \citet{henry2005general,henry2008analysis} express the vanilla options price under the normal SABR model with a two-dimensional integral, although it is later corrected by \citet{korn2013exact}. \citet{antonov2015mixing} further simplify the price to a one-dimensional integration with an approximation. However, in the absence of efficient numerical schemes to evaluate those integral representations, the normal volatility approximation in Equation~(\ref{eq:hagan})
remains as a practical approach to price vanilla options under the normal SABR model.

The development of MC simulation methods of the SABR dynamics is relatively recent. While several efficient approximations~\citep{chen2012low,leitao2017one,leitao2017multi} have been proposed, an exact simulation method~\citet{cai2017sabr} is available for the following three special cases: (i) $\rho=0$, (ii) $\beta=0$, and (iii) $\beta=1$. The key element in this method is to simulate the time-integrated variance $\ExpBM[-1/2]_S$, conditional on the terminal volatility $Z_S$. The cumulative distribution function (CDF) of the quantity is obtained from the Laplace transform of $(1/\ExpBM[-1/2]_S)\,|\,Z_S$, which has a closed-form expression~\citep{matsuyor2005exp1}.  Given the exact random numbers of $Z_S$ and $\ExpBM[-1/2]_S$, $X_{\ExpBM[-1/2]_S}$ in the normal SABR model is easily simulated as $X_1\sqrt{\ExpBM[-1/2]_S}$ for a standard normal variable $X_1$. Although a heavy Euler scheme is avoided, the method of \citet{cai2017sabr} still incurs a moderate computation cost due to the numerical inversion of the Laplace transform and root-solving for the CDF inversion.

As opposed to the previous studies using $\Hplane{2}$, this study's main result in section~\ref{sec:mc}, based on \citet{alili1997}, exploits $\Hplane{3}$. The pairs $(x_t, z_t)$ and $(y_t,z_t)$ from the standard BM in $\Hplane{3}$ are the two NSVh processes with $\drift=-1$ and $\rho=0$, and it should be noted that the introduction of $\drift$ to the NSVh model makes this connection possible. Despite a higher dimensionality, the heat kernel of $\Hplane{3}$ is given without an integral, and \citet{alili1997} shows that the squared radius $x_t^2+y_t^2$ and $z_t$ can be exactly simulated. Subsequently, $x_t$ (or $y_t$) is extracted via cosine (or sine) projection with random angle. Therefore, the study's simulation scheme can be considered as a hyperbolic-geometry extension of the Box-Muller algorithm~\citep{box1958note} for generating normal random variable, in which the $z$-axis of $\Hplane{3}$ is additionally added.

\subsection{Johnson's Distribution Family} \label{sec:jsu_pre}
\citet{johnson1949systems} proposes a system of distribution families in which a random variable $X$ is represented by the transformations from a standard normal variable $Z$:
\begin{equation} \label{eq:johnson}
\frac{X - \gamma_X}{\delta_X} = f \left(\frac{Z-\gamma_Z}{\delta_Z}\right) \quad\text{for}\quad 
f(x) = \begin{cases}
1/(1+e^{-x}) & \text{for $S_B$ (bounded) family} \\
e^x \quad & \text{for $S_L$ (lognormal) family} \\
\sinh x & \text{for $S_U$ (unbounded) family},
\end{cases}
\end{equation}
where $\gamma_X$ and $\gamma_Z$ are location parameters and $\delta_X$ and $\delta_Z$ are scaling parameters. Although not explicitly included, normal distribution can be considered as a special intersection of the three families in the limit of $\delta_X$ and $\delta_Z$ proportionally going to infinity. Therefore, it is 
often included as $S_N$ (normal) family with $f(x)=x$. The $X$ range is unbounded for $S_U$ and $S_N$, semi-bounded for $S_L$, and bounded for $S_B$. The system is designed in such a way that a unique family is chosen for any mathematically feasible pair of skewness and kurtosis. For a fixed value of skewness, the kurtosis increases in the order of $S_B$, $S_L$, and $S_U$.

Particularly, the $S_U$ family has been an attractive choice for modeling a heavy-tailed data set, and has been adopted in various fields; refer to \citet{jones2014web} and the references therein. Examples in finance includes heavy-tailed innovation in the GARCH model~\citep{choi2008asymmetric}, prediction of value-at-risk~\citep{simonato2011performance,venkataraman2016est}, and asset return distribution \citep{shang2004modeling,corlu2015modelling}.

The $S_U$ distribution has several advantages over alternative heavy-tailed distributions. First, it explains a wide range of skewness and kurtosis. For a fixed value of skewness, it can accommodate arbitrary high values of kurtosis, which is not feasible in the classical approaches to generalize a normal distribution, such as the Gram-Charlier or Cornish-Fisher expansions. Second, many properties of the distributions are available in closed forms: PDF, CDF, skewness, and kurtosis. Third, the parameters are efficiently estimated---refer to \citet{tuenter2001algo} for the moment matching in the reduced form and \citet{wheeler1980quantile} for the quantile-based estimation. Finally, drawing random numbers is easy, which makes $S_U$ distribution ideal for MC simulations, particularly in a multivariate setting~\citep{biller2006multi}. In general, random number sampling is not trivial, even if the distribution functions are given in closed forms.

In addition to the existing merits, our result in section~\ref{sec:jsu} gives a first-class-citizen status to the $S_U$ distribution among other heavy-tailed distributions by showing that it is a solution of a continuous-time SV process, the NSVh model with $\drift=1$. This partially explains why the $S_U$ distribution has been superior in modeling asset return distributions and risk metrics. 

\section{Main Results} \label{sec:main}
The study's main results first present  \citet{bougerol1983}'s identity in the original form. Since the original identity is generalized later, we state it as a Corollary and defer the proof to Proposition~\ref{prop:boug2}.
\begin{coro}[Bougerol's identity]
	\label{coro:bougerol}
	For a fixed time $T$, the following is equal in distribution:
	\begin{equation}
	\int_0^T e^{Z_t} dX_t \;\distequal\; X_{\ExpBM[0]_T} \;\distequal\; \sinh(W_T),
	\end{equation}
	where $X_t$, $W_t$, and $Z_t$ are independent BMs, and $\ExpBM[0]_T$ is defined by Equation~(\ref{eq:ExpBM}).
\end{coro}
This identity is surprising in that the stochastic integral involving two independent BMs is equal in distribution to the $\sinh$ transformation of one BM. Refer to \citet{matsuyor2005exp1,vakeroudis2012bougerol} for a review and related topics. The identity should be interpreted with caution; the equality holds as distribution ($\distequal$) at a fixed time $t=T$, not as a process for $0\le t\le T$. Moreover, it does not directly help to solve Equation~(\ref{eq:F_S}). The identity must be generalized to non-zero drift, $\ExpBM[\mu]$, and provide the joint distribution with $\Zdrift[\mu]$, which is found in \citet{alili1997} and \citet{alili1997iden}. In the following subsections, we apply two generalizations to the NSVh model; one to the general $\drift$, and the other to a special case $\drift=1$.

\subsection{Monte-Carlo Simulation Scheme} \label{sec:mc}
The first generalization makes use of the BMs in $\Hplane{3}$ to obtain the distribution of $X_{\ExpBM_T}$, conditional on $\Zdrift_T$. We restate Proposition 3 in \citet{alili1997} in a modified and stronger form:
\begin{restatable}[\textbf{Bougerol's identity in hyperbolic geometry}]{prop}{boughyp}
\label{prop:boug1} Let $X_t$ and $Z_t$ be two independent BMs and the function $\phi$ defined by
\begin{equation}
\phi(Z,D) = e^{Z/2} \sqrt{2\cosh D-2\cosh Z} \qtext{for} Z\le D,
\end{equation}
then the following is equal in distribution, conditional on $\Zdrift_T$:
\begin{equation}
\int_0^T e^{\Zdrift_t} dX_t \distequal X_{\ExpBM_T} \distequal \cos\theta\; \phi\hspace{-0.25em}\left(\Zdrift_T, \sqrt{R_T^2 + (\Zdrift_T)^2}\right),
\end{equation}
where $R_t$ is a two-dimensional Bessel process, that is, the radius of a BM in two-dimensional Euclidean geometry, and $\theta$ is a uniformly distributed random angle. The three random variables $R_T$, $\Zdrift[0]_T$, and $\theta$ are independent.
\end{restatable}
Refer to Appendix~\ref{sec:append1} for the proof. Among the two original proofs in \citet{alili1997}, the proof exploiting the interpretation with $\Hplane{3}$ is provided. The proposition in this study is stronger than the original because the distribution equality holds, \textit{conditional on} $\Zdrift$, which is implied from the proof without difficulty. The study further simplifies the original proof by using the $\Hplane{3}$ heat kernel. Proposition \ref{prop:boug1} can be directly applied to Equation~(\ref{eq:F_S}) to obtain the transition equation and MC scheme:
\begin{coro} \label{coro:jointmc}
The joint distribution of the NSVh model at a fixed time $S$ is given as 
\begin{equation} \begin{aligned}
\nondim{\sigma}_S &= \exp\left(\Zdrift[(\drift-1)/2]_S\right) \qtext{and} \\
\label{eq:nsvh_price}
\nondim{F}_S &\distequal\; \rho\,\left( e^{\Zdrift[(\drift-1)/2]_S} -  e^{\drift S/2} \right) + \rhoc \cos\theta\, \phi\hspace{-0.25em}\left(\Zdrift[(\drift-1)/2]_S, \sqrt{R_S^2 + (\Zdrift[(\drift-1)/2]_S)^2}\right).
\end{aligned} \end{equation}
Furthermore, the three independent random variables can be simulated as
\begin{equation} \label{eq:R_cos}
\left( Z_S, R_S^2, \; \cos\theta \right) 
\distequal
\left( Z_1\sqrt{S},\; (X_1^2+Y_1^2)S,\; \frac{X_1 \;\text{(or $Y_1$)}}{\sqrt{X_1^2+Y_1^2}} \right),
\end{equation}
where $X_1$ and $Y_1$ are independent standard normals.
\end{coro}
The simulation method using standard normal variables in Equation~(\ref{eq:R_cos}) is more efficient than drawing $R_S$ and $\theta$ independently because the costly $\cos\theta$ evaluation is avoided. The idea is similar to the Marsaglia polar method~\citep{marsaglia1964polar} for drawing normal random variables. It must be noted that three random numbers---$X_1$, $Y_1$, and $Z_1$---generate two pairs of $\nondim{F}_S$ and $\nondim{\sigma}_S$. Therefore, one draw only requires one and a half (1.5) normal random variables, which is an unprecedented efficiency for any SV model simulation. Particularly, this method is more efficient than the exact SABR simulation of \citet{cai2017sabr}, although it is limited to the normal case. This study's method directly draws $X_{\ExpBM[(\drift-1)/2]_S}$ and $Z_S$, whereas \citet{cai2017sabr} first draws $\ExpBM[1/2]_S$ and $Z_S$, and subsequently $X_{\ExpBM[-1/2]_S}$. Although Equation~(\ref{eq:nsvh_price}) states the transition from $s=0$ to $s=S$, it can handle any time interval from $s=S_1$ to $s=S_2$ ($S_1<S_2$). Therefore, the scheme is ideal for pricing path-dependent claims.

\subsection{$S_U$ Distributions for $\drift=1$} \label{sec:jsu}
This subsection shows that the NSVh distribution for $\drift=1$ is expressed by the $S_U$ distribution and is related to Bougerol's identity generalized to an arbitrary starting point. In the following proposition, Proposition 4 of \citet{alili1997} (or Theorem 3.1 of \citet{matsuyor2005exp1}) is restated. More general results are found in Proposition 1 of \citet{alili1997iden} (or Proposition 2.1 of \citet{vakeroudis2012bougerol}), and we follow the proof therein.
\begin{prop}[\textbf{Bougerol's identity with an arbitrary starting point}] \label{prop:boug2}
For a fixed time $T$ and independent BMs, $X_t$, $Z_t$, and $W_t$, the following is equal in distribution:
\begin{equation} \label{eq:prop3}
\sinh(a)\,e^{Z_T}+\int_0^Te^{Z_t}dX_t \;\distequal\; 
\sinh(a)\,e^{Z_T}+ X_{\ExpBM[0]_T} \distequal \sinh(W_T + a).
\end{equation}
\end{prop}
\begin{proof}
The two processes,
$$ P_t = \sinh(W_t+a) \quad\text{and}\quad Q_t = e^{Z_t} \left( \sinh(a) + \int_0^t e^{-Z_s}dX_s \right),$$
are equivalent because they start from the same starting point $P_0 = Q_0 = \sinh(a)$ and follow the SDE:
$$ dP_t = \frac12 P_t dt + \sqrt{1+P_t^2}\;dW_t \quad\text{and}\quad
dQ_t = \frac12 Q_t dt + dX_t + Q_t\;dZ_t = \frac12 Q_t dt + \sqrt{1+Q_t^2}\;dW_t.
$$
Therefore, $P_t$ and $Q_t$ have the same distribution for any time $t$. The equality between $Q_T$ and the left-most expression is shown by the time-reversal $s \rightarrow T-s$. For a fixed time $T$, $Z_T-Z_{T-s}$ for $0\le s\le T$ is also a standard BM with the same ending point $Z_T$, and therefore it may be replaced with $Z_s$.
\end{proof}
The original Bougerol's identity in Corollary \ref{coro:bougerol} is a special case, with $a=0$. Now, Proposition \ref{prop:boug2} can be applied to further simplify the NSVh distribution for $\drift=1$.
\begin{coro} \label{coro:jsu}
The price of the NSVh model with $\drift=1$ at a fixed time $S$ follows a re-parametrized $S_U$ distribution: 
\begin{equation}\label{eq:FT_nsvh}
\nondim{F}_S \distequal\; \rhoc \sinh\left( W_S + \atanh\, \rho\right) - \rho\, e^{S/2},
\end{equation}
where the original parameters are mapped by 
$$
\delta_Z = \frac{1}{\sqrt{S}}, \quad
\frac{\gamma_Z}{\delta_Z} = -\atanh\,\rho, \quad
\delta_X = \frac{\sigma_0 \rhoc}{\alpha}, \quad \text{and}\quad
\gamma_X = \fwd - \frac{\sigma_0 \rho}{\alpha} e^{S/2}.
$$
It also admits a simpler form:
\begin{equation}\label{eq:FT_nsvh2}
\nondim{F}_S \distequal \sinh (W_S) + \rho\,\left(\cosh (W_S) - e^{S/2}\right).
\end{equation}
\end{coro}
\begin{proof}
The results are easily proved from the following hyperbolic function identities,
$$\asinh\left(\frac{\rho}{\rhoc}\right) = \text{atanh}\,\rho = \frac12\log\left( \frac{1+\rho}{1-\rho}\right).$$
\end{proof}
Although Proposition \ref{prop:boug2} is a well-known result, to the best of our knowledge, this is the first time that it is interpreted in the context of the $S_U$ distribution or SV model. 
Compared to Corollary~\ref{coro:jointmc}, Corollary \ref{coro:jsu} is an even more efficient MC scheme requiring one normal random number for one draw of price, although for a special case $\drift=1$. However, this is at the expense of the terminal volatility $\nondim{\sigma}_S$ being lost. Unlike Corollary \ref{coro:jointmc}, Corollary \ref{coro:jsu} can only generate the final price $\nondim{F}_S$, and therefore cannot be used for path-dependent claims. 

The key consequence of Corollary~\ref{coro:jsu} is that the NSVh model bridges the SDE-based SV model and heavy-tailed distribution. The encounter between the two subjects is mutually beneficial. First, as a solution of the NSVh process, the $S_U$ distribution obtains a more intuitive parametrization than the original in Equation~(\ref{eq:johnson}). From Equation~(\ref{eq:FT_nsvh2}), it is clear that the symmetric heavy tail comes from the $\sinh$ term, controlled by $S$, and the asymmetric skewness from the $\cosh$ term, controlled by $\rho$. The new parametrization also helps to understand the relationship between Johnson family members. The lognormal family $S_L$ is recognized as a special case with $\rho=\pm 1$ ($\rhoc=0$) as $\nondim{F}_S \; \distequal\; \pm\, ( e^{\Zdrift[0]_S}-e^{S/2} )$. The normal family $S_N$ is obtained as $\nondim{F}_S/\sqrt{S}$ in the limit of $S\rightarrow 0$. Under the NSVh parameters, the well-known PDF and CDF of the $S_U$ distribution are respectively expressed by 
\begin{equation} \label{eq:jsu_cdf}
\begin{gathered}
p_{\drift=1}(x) = \frac{n(d)}{\rhoc\sigma_0\sqrt{T}\sqrt{1+\xi^2}} \quad
\text{and}\quad	P_{\drift=1}(x) = N(-d) \\
\text{where}\quad 
	d = \frac{1}{\sqrt{S}}\left( \asinh\left( 
\frac{\alpha}{\rhoc\sigma_0}(\fwd-x) - \frac{\rho}{\rhoc} e^{S/2} 
\right) + \atanh\,\rho \right),
\end{gathered}
\end{equation}

Conversely, from the $S_U$ distribution, the NSVh model obtains analytic tractability for option price and risk measures. Below are the closed-form expressions for the quantities of interest:
\begin{coro} For an asset price following the NSVh process with $\drift=1$, option price, value-at-risk, and expected shortfall have the following closed-form solutions.
\begin{itemize}
	\item The undiscounted price of a vanilla option with strike price $x$:
	\begin{equation} \label{eq:jsu_opt}
	V_\pm(x) = \frac{\sigma_0}{2\alpha} e^{S/2} \left((1+\rho)N(d+\sqrt{S}) - (1-\rho)N(d-\sqrt{S})-2\rho N(d)\right) \pm \left(\fwd-K\right) N(\pm d),
	\end{equation}
	where $d$ is defined in Equation~(\ref{eq:jsu_cdf}) and $\pm$ indicates call/put options, respectively. 
	\item Value-at-risk for quantile $p$:
	\begin{equation} \label{eq:jsu_var}
	\text{VaR}(p) = \fwd - \frac{\sigma_0}{\alpha} 
	\left( \rhoc\, \sinh\left( d\sqrt{S} - \atanh\,\rho \right) + \rho\, e^{S/2} \right) \quad\text{for}\quad d = -N^{-1}(p).
	\end{equation}
	\item Expected shortfall for quantile $p$:
	\begin{equation} \label{eq:jsu_es}
	\text{ES}(p) = \fwd - \frac{\sigma_0 e^{S/2}}{2\alpha p} \Big((1+\rho)N(d+\sqrt{S}) - (1-\rho)N(d-\sqrt{S}) - 2\rho(1-p)\Big).
	\end{equation}
\end{itemize}
\end{coro}
\begin{proof}
The option prices are easily derived by integrating Equation~(\ref{eq:FT_nsvh2}) with the boundary $d$ obtained from Equation~(\ref{eq:FT_nsvh}). The relationship between the put option value and the expected shortfall, $\text{ES}(p) = \text{VaR}(p) - V_-(\text{VaR}(p)) / p$, is useful, where $\text{VaR}(p)$ vanishes in the final expression for $\text{ES}(p)$.
\end{proof}
Later, it is argued that the NSVh distributions with different values of $\drift$ are close to each other, and thus the analytically tractable $\drift=1$ case can represent the rest including the normal SABR model ($\drift=0$). The option price from the closed-form formula serves as a benchmark against which the option price from the MC scheme of Corollary~\ref{coro:jointmc} is compared in Section~\ref{sec:num}.

\subsection{Moments Matching of the NSVh Distribution} \label{sec:mom}
The study derives the moments of the NSVh distribution for general $\drift$ to be used for parameter estimation. The study also proposes a moment matching in the reduced form for $\drift=0$ to complement that for $\drift=1$ by \citet{tuenter2001algo}.

\begin{restatable}{coro}{nsvhmoments}
\label{coro:mom}
The central moments of the canonical NSVh distribution $\nondim{\mu}_n = E(\nondim{F}_S^{\;n})$, for $2\le n\le 4$, are given as
\begin{equation}
\begin{aligned}
\nondim{\mu}_2 =& \rho^2\; \ExpS^{\drift} (\ExpS-1) + \rhoc^2\; \frac{\ExpS^{1+\drift} - 1}{1+\drift} \qtext{for} w=e^S\ge 0, \\
\nondim{\mu}_3 =& \rho^3\; \ExpS^{\frac32\drift } (\ExpS-1)^2(\ExpS + 2) + 3\rho\rhoc^2 \;\ExpS^{\frac12 \drift } \left(\frac{\ExpS^{3+\drift}-1}{3+\drift} - \frac{\ExpS^{1+\drift}-1}{1+\drift} \right),\quad\text{and} \\ 
\nondim{\mu}_4 =& \rho^4\; \ExpS^{2\drift } (\ExpS-1)^2(\ExpS^4+2\ExpS^3+3\ExpS^2-3) + 6\rho^2\rhoc^2 \;\ExpS^{\drift} \left( \ExpS\;\frac{\ExpS^{5+\drift}-1}{5+\drift} - 2\;\frac{\ExpS^{3+\drift}-1}{3+\drift} \right. \\
& \left. + \frac{\ExpS^{1+\drift}-1}{1+\drift} \right) + \frac{3}{2} \rhoc^4\; \left( - \ExpS^{1+\drift}\;\frac{\ExpS^{5+\drift}-1}{5+\drift} + (\ExpS^{3+\drift}+1)\frac{\ExpS^{3+\drift}-1}{3+\drift} - \frac{\ExpS^{1+\drift}-1}{1+\drift} \right).
\end{aligned}
\end{equation}
The central moments of the original form can be scaled as $\mu_n = E((F_T-\fwd)^n)= (\sigma_0/\alpha)^n \nondim{\mu}_n$, and the skewness and ex-kurtosis are given as $s = \nondim{\mu}_3 / \nondim{\mu}_2^{3/2}$ and $\kappa = \nondim{\mu}_4/\nondim{\mu}_2^2-3$, respectively. For the normal SABR ($\drift=0$), further simplified expressions are obtained:
\begin{equation} \label{eq:mom_nsabr}
\nondim{\mu}_2 = w-1, \quad s = \rho (w+2)\sqrt{w-1}, \quad \kappa = (w-1) \left( \left(\frac{4\rho^2+1}5\right)(w^3+3w^2+6w+5) +1\right).
\end{equation}
\end{restatable}
Refer to Appendix~\ref{sec:append2} for detailed derivation. Corollary~\ref{coro:mom} generalizes the moments for $S_L$ ($\rho=\pm 1$) and $S_U$ ($\drift=1$) distributions. 

The similarity of the NSVh distributions for different $\drift$ is inferred from skewness and ex-kurtosis. The result $\nondim{\mu}_k=O(w^{k(\drift+k-1)/2})$ for large $w$, at least for $k=2,3,$ and 4 implies that the leading order of skewness and ex-kurtosis are independent of $\drift$ as $s=O(w^{3/2})$ and $\kappa=O(w^4)$, as indicated in Equation~(\ref{eq:mom_nsabr}). To illustrate that, in Figure~\ref{fig:skew_kurt}, the contours of skewness and ex-kurtosis for $\drift=0$ and 1, as functions of $S$ (ex-kurtosis) and $\rho S$ (skewness) for $\rho\ge 0$, are plotted. Although the parameters for $\drift=0$ is slightly higher than those for $\drift=1$ to obtain the same skewness and ex-kurtosis levels, the contours for $\drift=0$ and $1$ are very similar implying the similarity between the normal SABR and the $S_U$ distributions. 
\begin{figure}
\caption{\label{fig:skew_kurt}
Contour Plot of Skewness (Red Dashed Line) and Excess Kurtosis (Blue Solid Line) for Varying $S\,(=\alpha^2 T)$ versus $\rho S$. The upper left triangle $(\rho S, S)$ is for $\drift=1$ ($S_U$) and the lower right triangle $(S, \rho S)$ for $\drift=0$ (normal SABR). The values for skewness are 0, 1.5, 3, 4.5, 6, and 8, and the values for excess kurtosis are 2, 7, 16, 40, 100, and 200 from the lower left to the upper right corner.}
\centering
\includegraphics[height=0.55\textwidth]{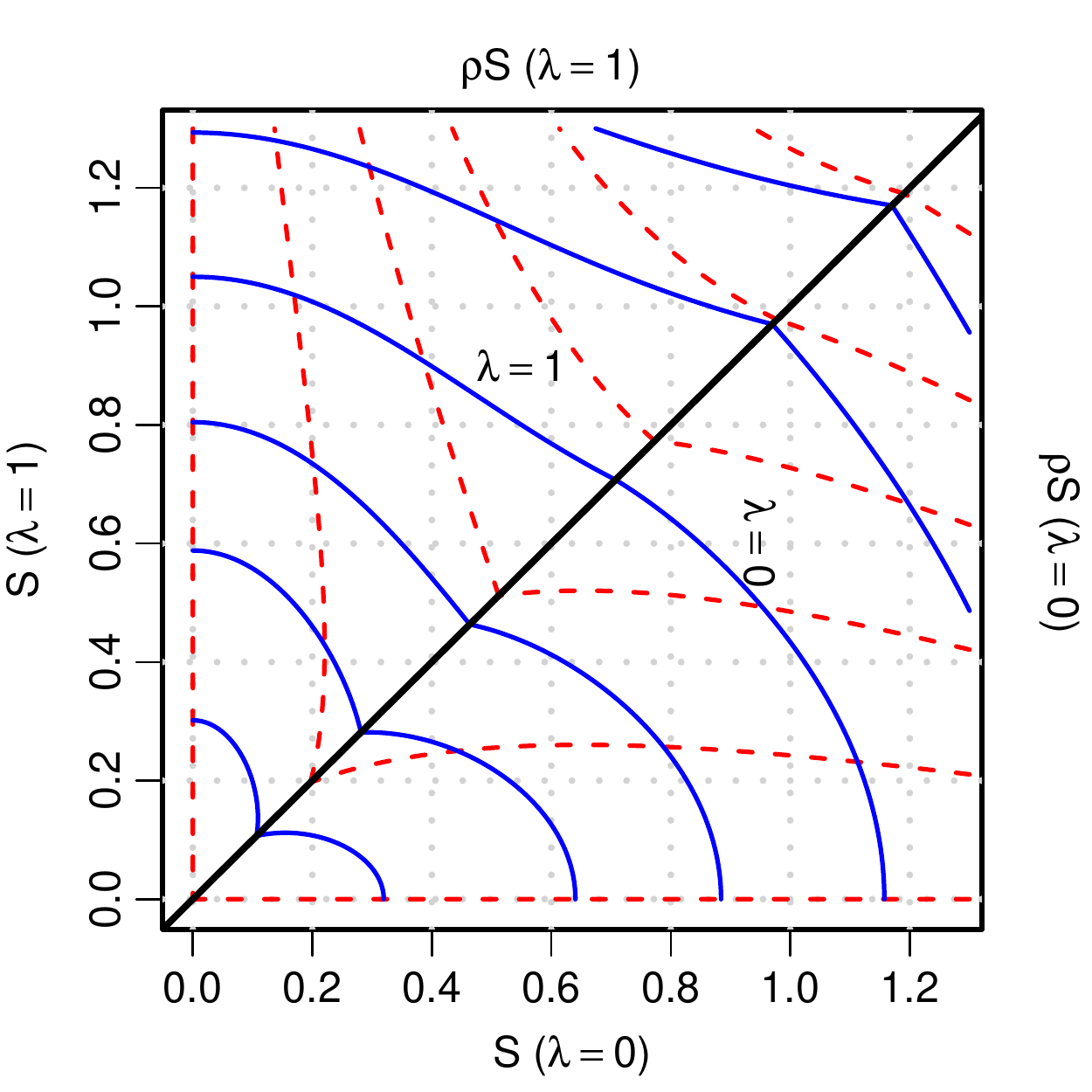}
\end{figure}

Parallel to \citet{tuenter2001algo}'s reduced moment matching method for the $S_U$ distribution, a similar method is developed for the normal SABR model. Combined with \citet{tuenter2001algo}, the two methods can quickly find an equivalent parameter set of one distribution from the other. By joining the expressions for $s$ and $\kappa$ in Equation~(\ref{eq:mom_nsabr}) through $\rho$, $\kappa$ is expressed as a univariate function of $w\ge 1$:
\begin{equation} \label{eq:mom_match}
f(w) = \frac{4s^2(w^3+3w^2+6w+5)}{5(w+2)^2} + (w-1)\left(1 + \frac{1}{5} (w^3+3w^2+6w+5)\right),
\end{equation}
for which the study numerically finds the root $w_*$ of $\kappa = f(w_*)$. It can be shown that $f(w)$ is monotonically increasing for $w\ge 1$. Therefore, the root $w_*$ would be unique if it exists. We can further bound $w_*$ by $w_m\le w_*\le w_M$ to expedite the numerical root-finding. Lower bound $w_m$ is the unique cubic root of $s^2 = (w-1)(w+2)^2$ (the $\rho=\pm 1$ case) for $w\ge 1$:
$$ w_m = 2\cosh\left(\frac13 \acosh \left( 1+\frac{s^2}{2}\right) \right).$$
Upper bound $w_M$ is obtained by plugging $w=w_m$ into Equation~(\ref{eq:mom_match}), except the $(w-1)$ term:
$$ w_M = 1+\frac{\kappa - \frac45 s^2(w_m^3 + 3w_m^2+6w_m+5)/(w_m+2)^2}{1+\frac15 (w_m^3+3w_m^2+6w_m+5)}.$$
The existence of $w_*$ is equivalent to $f(w_m) \le \kappa$. If $w_*$ exists and is found from the numerical root-finding, then the parameters can be solved as
$$ S = \log w_*, \quad \rho = \frac{s}{(w_*+2)\sqrt{w_*-1}}, \quad \text{and}\quad \frac{\sigma_0}{\alpha} = \sqrt{\frac{\mu_2\log w_*}{(w_*-1)S}}.$$

\subsection{Summary of results}
In Figure~\ref{fig:venndiagram}, the relationship of the NSVh model and other related models is shown.
In Table~\ref{tab:summary}, the results for the three important drift values, $\drift=-1$, 0, and 1, are summarized for comparison.

\begin{figure}
\caption{\label{fig:venndiagram}
The Overview of the NSVh Model in Relation to Other Previously Known Models and Distributions.}
\centering
\includegraphics[width=0.98\linewidth]{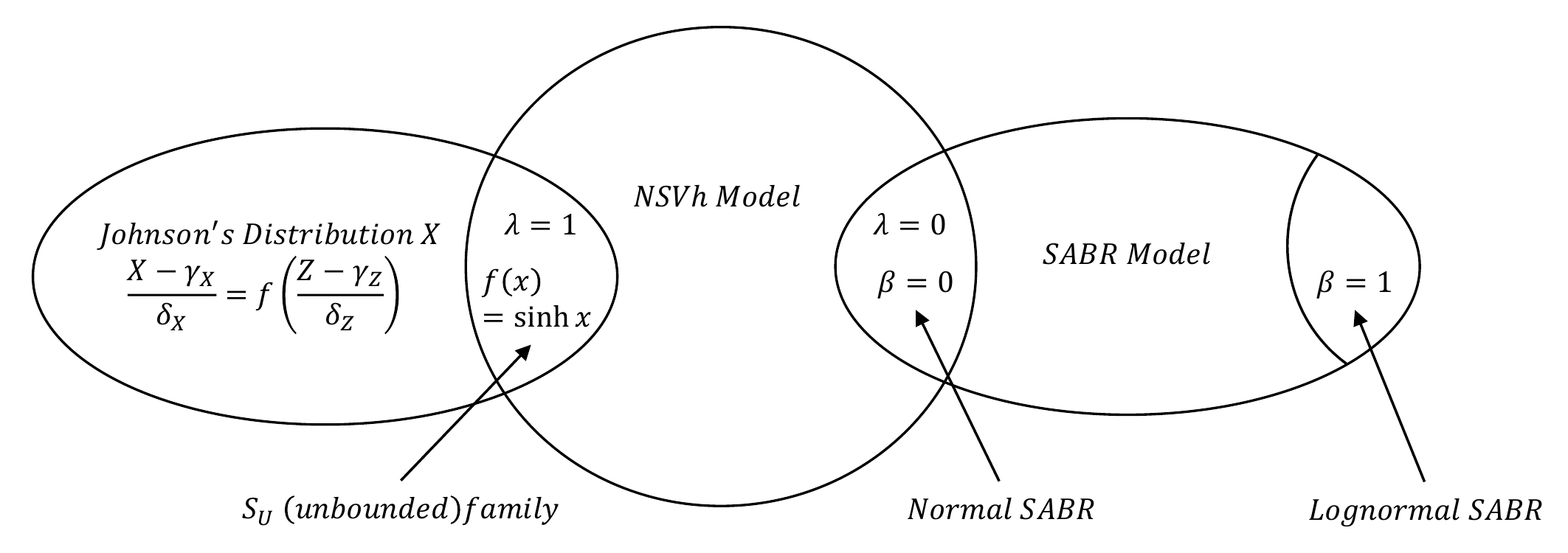}
\end{figure}

\begin{table}
\caption{Summary of the NSVh Model for the Three Key Special Cases: $\drift=-1,0$, and $1$.}
\label{tab:summary}
\begin{center}
\begin{tabular}{|c||c|c|c|} \hline
	Drift parameter	& $\drift=-1$ &  $\drift=0$ & $\drift=1$ \\ \hline \hline
	Equivalent& \multirow{2}{*}{Standard BM in $\Hplane{3}$} & Normal SABR & \multirow{2}{*}{$S_U$ distribution} \\
	model/distribution&  & Standard BM in $\Hplane{2}$ & \\ \hline
	Drifted BM for $\nondim{\sigma}_s$ & $\Zdrift[-1/2]_s$ & $\Zdrift[0]_s$ & $\Zdrift[1/2]_s$ \\ \hline 
	Terminal volatility $\nondim{\sigma}_S$ & $\exp(\Zdrift[-1]_S)$ & $\exp(\Zdrift[-1/2]_S)$ & $\exp(\Zdrift[0]_S)$ \\ \hline
	Integrated variance & $\ExpBM[-1]_S$ & $\ExpBM[-1/2]_S$ & $\ExpBM[0 ]_S$ \\ \hline
	Exact MC simulation & \multicolumn{2}{c|}{Corollary \ref{coro:jointmc}} & Corollary \ref{coro:jointmc} \& \ref{coro:jsu} \\ \hline
	\multirow{2}{*}{Vanilla option price} &  & Equation~(\ref{eq:hagan}) & Equation~(\ref{eq:jsu_opt}) \\
	 &  & (approximation) & (exact) \\ \hline
	Moments & \multicolumn{3}{c|}{Corollary \ref{coro:mom}} \\ \hline
	Moment matching & & Equation~(\ref{eq:mom_match}) &\citet{tuenter2001algo} \\ \hline
\end{tabular}
\end{center}
\end{table}

\section{Parameter Estimation from Empirical Data} \label{sec:num}
The NSVh distribution is calibrated to two empirical data sets--- 
swaption volatility smile and daily stock index return. The purpose of this exercise is to demonstrate various numerical procedures presented in this study rather than arguing that the NSVh model is superior to other SV models or heavy-tailed distributions in fitting these data. Additionally, the study shows that the two NSVh models, that is, $\drift=0$ (normal SABR) and $\drift=1$ ($S_U$), yield very similar distributions, and thus can be used interchangeably, if calibrated to the same target, such as implied volatility or moments.

\subsection{Swaption Volatility Smile}
The study obtains the US swaption market prices on March 14, 2017 from Reuters. The two heavily traded expiry--tenor pairs of the US swaption---1y1y and 10y10y---are chosen to illustrate different volatility smile shapes. To avoid the complication of the annuity price of the underlying swap, the study computes the price in the unit of the annuity from the BSM implied volatilities provided by Reuters, rather than raw dollar prices. Refer to the insets of Figure~\ref{fig:usd_swo} for the BSM implied volatilities.

\begin{figure}
\caption{\label{fig:usd_swo}
Swaption Volatility Smile in the Implied Normal Volatility (in annual basis points) Observed in the US Market on March 14, 2017: (a) 1 year into 1-year swap (1y1y) and (b) 10 years into 10-year swap (10y10y). The circles represent the volatilities implied from observed market prices, among which the black ones are ATM and ATM $\pm$ 1\% points for calibration. The solid (blue) line represents the model-implied smile curve for $\drift=1$ ($S_U$) and the dashed (red) one for $\drift=0$ (normal SABR). The insets are the implied BSM volatilities (in annual \%).
}
\centering
\includegraphics[width=0.49\textwidth]{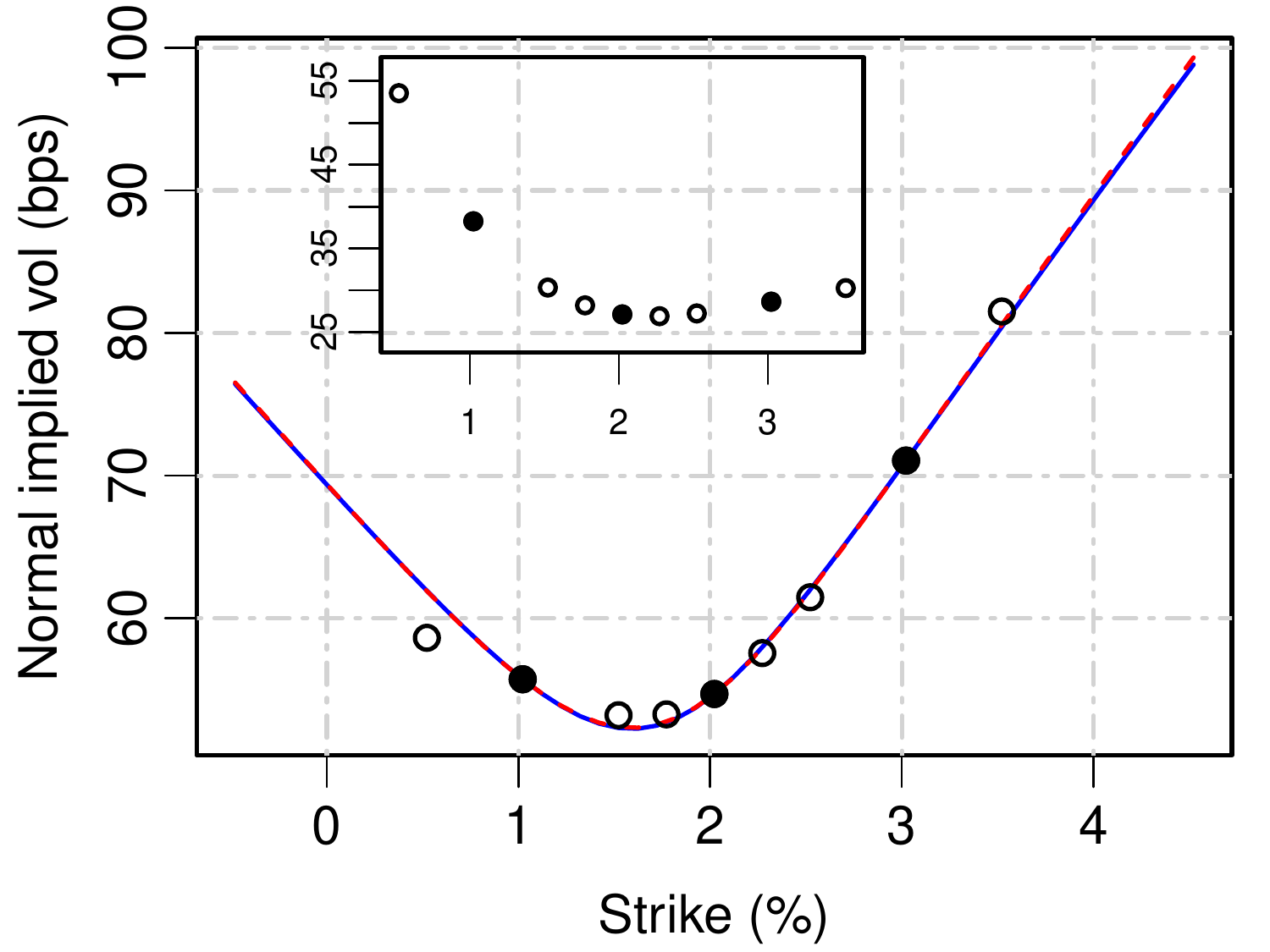} 
\includegraphics[width=0.49\textwidth]{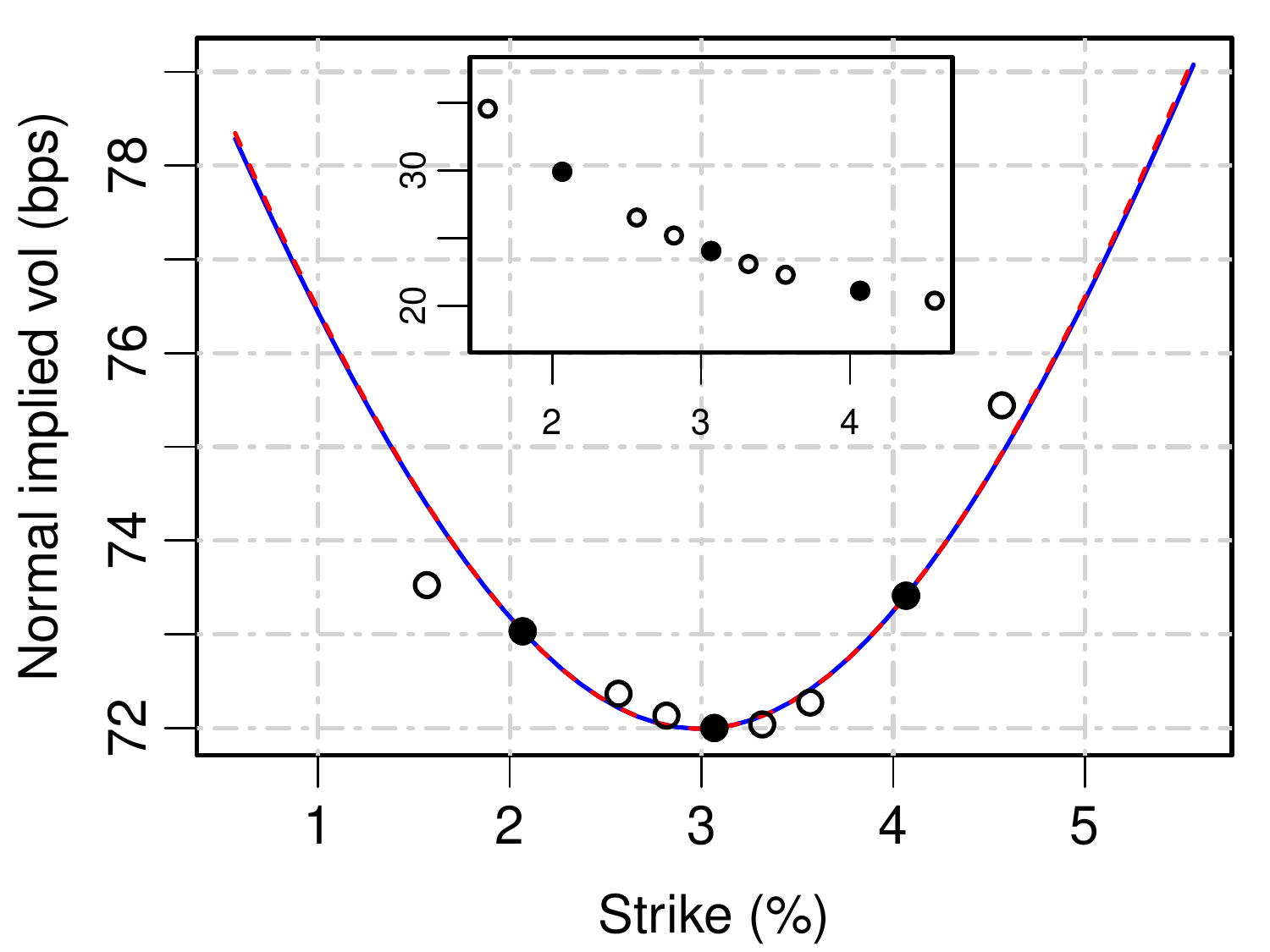}
\end{figure}

Figure~\ref{fig:usd_swo} shows the normal volatility smile implied from the market and the calibrated models. While option prices are observable for strike prices with spreads of 0, $\pm 0.25\%$, $\pm 0.5\%$, $\pm 1.0\%$, and $\pm 1.5\%$ from the forward swap rates $\fwd$, the study uses only the three spreads---0 and $\pm 1\%$ for calibration---to ensure that the calibrated parameter set---$(\sigma_0, \alpha, \rho)$---reproduces option prices at those three strike prices. For calibration, the study uses Equations~(\ref{eq:hagan}) for $\drift=0$ and (\ref{eq:jsu_opt}) for $\drift=1$. The volatility smile curves implied from the two models are indistinguishable, and thus close to each other in distributions. Table~\ref{tab:tab1} shows the calibrated parameter values for $\drift=0$ and 1 for reference.

In Table~\ref{tab:tab2}, the option prices are compared from the MC simulation of Corollary~\ref{coro:jointmc} with the option prices from the aforementioned analytic methods. In the experiment, the MC simulation of $10^6$ samples is repeated 100 times. For $\drift=1$, both Equation~(\ref{eq:jsu_opt}) and the MC method are exact, and thus the option prices from the two methods have very little difference due to the MC noise. For $\drift=0$, however, Equation~(\ref{eq:hagan}) is an approximation. The analytic prices show clear deviation from the accurate MC prices. Overall, this exercise reconfirms the possibility of the $S_U$ distribution being used as a better alternative to the normal SABR model.

\begin{table}
\caption{Parameters Calibrated to the US Swaption Volatility Smile on March 14, 2017. Refer to Figure~\ref{fig:usd_swo} for the calibration points. Equations~(\ref{eq:hagan}) and (\ref{eq:jsu_opt}) are used for the NSVh models with $\drift=0$ and $\drift=1$, respectively.
\label{tab:tab1}}
\begin{center}
\begin{tabular}{|c||r|r||r|r|} \hline
Calibrated	& \multicolumn{2}{c||}{1y1y ($T=1$)} & \multicolumn{2}{c|}{10y10y ($T=10$)} \\ \cline{2-5}
Parameters & \multicolumn{1}{c|}{$\drift=0$} & \multicolumn{1}{c||}{$\drift=1$} & \multicolumn{1}{c|}{$\drift=0$} & \multicolumn{1}{c|}{$\drift=1$} \\ \hline
$\rho\;(\%)$ & 33.503 & 32.244 & 1.697 & 1.580\\
$\alpha\;(\%)$ & 61.962 & 62.181 & 22.372 & 22.196\\
$\sigma_0\;(\%)$ & 0.533 & 0.477 & 0.691 & 0.609\\
\cline{2-5}
$\fwd\;(\%)$ & \multicolumn{2}{c||}{2.0221} & \multicolumn{2}{c|}{3.0673}\\
\hline
\end{tabular}
\end{center}
\end{table}

\begin{table}
\caption{Vanilla Options Pricing Tested Against the Parameters of Table~\ref{tab:tab1} for 10y10y Swaption. Analytic prices ($P_\textsc{ana}$) are computed from Equations~(\ref{eq:hagan}) and (\ref{eq:jsu_opt}) and MC price ($P_\textsc{mc}$) is shown relative to $P_\textsc{ana}$ with the standard deviation. The MC simulation with $10^6$ paths is repeated 100 times. Prices are in the unit of the annuity of the underlying swap.
\label{tab:tab2}}
\begin{center}
	\begin{tabular}{|c||c|c||c|c|} \hline
	$K-\fwd$ & \multicolumn{2}{c||}{$\drift=0$} & \multicolumn{2}{c|}{$\drift=1$} \\ \cline{2-5}
	(bps) & $P_\textsc{ana}$ & $P_\textsc{mc}-P_\textsc{ana}$ & $P_\textsc{ana}$ & $P_\textsc{mc}-P_\textsc{ana}$ \\ \hline \hline
-200 & 2.275E-2 & -4.1E-5 $\pm$ 1.8E-5 & 2.274E-2 & 6.4E-7 $\pm$ 1.8E-5\\
-100 & 1.506E-2 & -2.1E-5 $\pm$ 1.6E-5 & 1.506E-2 & 6.0E-7 $\pm$ 1.6E-5\\
0 & 9.083E-3 & -1.2E-5 $\pm$ 1.3E-5 & 9.083E-3 & 5.3E-7 $\pm$ 1.3E-5\\
100 & 5.108E-3 & -2.6E-5 $\pm$ 1.1E-5 & 5.108E-3 & 3.3E-7 $\pm$ 1.1E-5\\
200 & 2.807E-3 & -4.8E-5 $\pm$ 8.8E-6 & 2.804E-3 & 3.9E-7 $\pm$ 9.1E-6\\
300 & 1.567E-3 & -6.0E-5 $\pm$ 7.0E-6 & 1.559E-3 & 3.7E-7 $\pm$ 7.3E-6\\
	\hline
\end{tabular} \\ \vspace{1em}
\end{center}
\end{table}

\subsection{Daily Return of Stock Index}
We fit the NSVh distribution to the daily returns of two stock indices---the US Standard \& Poor's 500 Index (S\&P 500) and China Securities Index 300 (CSI 300). The data covers the 12-year period from the beginning of 2005 to the end of 2016. For the analysis to be easily reproducible, daily returns are computed as the return on outright index values, rather than as holding period return. 

The statistics of the daily returns are summarized in Table~\ref{tab:tab4}. S\&P 500 shows heavier tails but less skewness than CSI 300. The table also shows the fitted parameters of the NSVh distributions for $\drift=0$ and 1. In fitting, the reduced moment-matching methods of \citet{tuenter2001algo} and section~\ref{sec:mom} are used. Based on these values, the value-at-risk and expected shortfall are computed and compared to those from the normal distribution assumption and the historical data. While Equations~(\ref{eq:jsu_var}) and (\ref{eq:jsu_es}) are used for $\drift=1$, MC simulation is used to compute the risk measures for $\drift=0$. The results are shown in Table~\ref{tab:tab5}. As the return distribution has heavy tail, the value-at-risk and expected shortfall from the NSVh distributions are much closer to those from the historical data than those from the normal distribution. 
The risk measures from $\drift=0$ and 1 are very close to each other, confirming the similarity between the two distributions.

\begin{table}
\caption{Summary Statistics and Fitted Parameters of Daily Returns of S\&P 500 and CSI 300 Indices from 2005 to 2016. The number of samples, mean, variance, skewness, and excess kurtosis are denoted by $n$, $\fwd$, $\mu_2$, $s$, and $\kappa$, respectively. The mean and variance are computed from percentage returns. We assume $T=1$ for convenience. \label{tab:tab4}}
\begin{center}
\begin{tabular}{|c||r|r||r|r|} \hline
Summary statistics & \multicolumn{2}{c||}{S\&P 500} & \multicolumn{2}{c|}{CSI 300} \\ \hline \hline
$n$ & \multicolumn{2}{c||}{3020} & \multicolumn{2}{c|}{2914} \\
$\fwd$ & \multicolumn{2}{c||}{0.0282} & \multicolumn{2}{c|}{0.0417} \\
$\mu_2$ & \multicolumn{2}{c||}{1.5154} & \multicolumn{2}{c|}{3.4092} \\
$s$ & \multicolumn{2}{c||}{-0.0933} & \multicolumn{2}{c|}{-0.5075} \\
$\kappa$ & \multicolumn{2}{c||}{11.4454} & \multicolumn{2}{c|}{3.3348} \\
\hline \hline
Fitted parameters & \multicolumn{1}{c|}{$\drift=0$} & \multicolumn{1}{c||}{$\drift=1$} & \multicolumn{1}{c|}{$\drift=0$} & \multicolumn{1}{c|}{$\drift=1$} \\ \hline
$\rho\;(\%)$ & -2.042 & -1.725 & -20.454 & -18.539 \\
$\alpha\;(\%)$ & 88.533 & 84.587 & 63.782 & 61.853 \\
$\sigma_0\;(\%)$ & 99.915 & 82.538 & 166.213 & 150.167 \\
\hline
\end{tabular}
\end{center}
\end{table}

\begin{table}
\caption{Value-at-risk (VaR) and Expected Shortfall (ES) from the Normal Distribution (Normal), the Two NSVh Distributions ($\drift=0$ and $\drift=1$), and the True Values from the Dataset (Sample).\label{tab:tab5}}
\begin{center}
\begin{tabular}{|c||c|c|c|c||c|c|c|c|} \hline
	\multirow{2}{*}{Risk Measures}	& \multicolumn{4}{c||}{S\&P 500} & \multicolumn{4}{c|}{CSI 300} \\ \cline{2-9}
& Normal & $\drift=0$ & $\drift=1$ & Sample & Normal & $\drift=0$ & $\drift=1$ & Sample \\ \hline
VaR ($p=5\%$) & -1.997 & -1.825 & -1.824 & -1.832 & -2.995 & -3.032 & -3.036 & -3.007 \\
VaR ($p=1\%$) & -2.836 & -3.405 & -3.432 & -3.615 & -4.254 & -5.234 & -5.246 & -5.732 \\ \hline
ES ($p=5\%$) & -2.511 & -2.857 & -2.872 & -3.042 & -3.767 & -4.433 & -4.440 & -4.745 \\
ES ($p=1\%$) & -3.253 & -4.781 & -4.820 & -5.309 & -4.879 & -6.849 & -6.857 & -7.298 \\
\hline 
\end{tabular}
\end{center}
\end{table}

Finally, the goodness-of-fit of the $S_U$ ($\drift=1$) is illustrated by using the probability plot. In Figure~\ref{fig:index}, the theoretical $Z$-score---$Z_0^{(j)} = N^{-1}((j-1/2)/n)$ for the $j$-th ordered sample---is shown on the $x$-axis and two different sample's $Z$-scores on the $y$-axis as follows: (i) from the estimated normal distribution $Z_1^{(j)}=(X_j-\fwd)/\sqrt{\mu_2}$ and (ii) from the $S_U$ distribution computed as
$$
Z_2^{(j)} = N^{-1}(P_{\drift=1}(X_j)) = \frac{1}{\sqrt{S}}\left( \asinh\left( 
\frac{\alpha}{\rhoc\sigma_0}(X_j-\fwd) + \frac{\rho}{\rhoc} e^{S/2} \right) - \atanh\,\rho \right).
$$
Therefore, $(Z_0, Z_2)$ is understood as the \textit{$S_U$ probability plot}, while $(Z_0, Z_1)$ is the normal probability plot in the usual definition. Figure~\ref{fig:index} shows that the points under the $S_U$ probability plot are close to the $y=x$ line, indicating that the stock returns closely follow the $S_U$ distribution ($\drift=1$).

\begin{figure}
\caption{\label{fig:index}
Probability Plots of Daily returns of the S\&P 500 (left) and CSI 300 Indexes (right). The $Z$-scores under the normal distribution (black dot) and the $S_U$ distribution (red circle) in $y$-axis are plotted against the theoretical $Z$-score in the $x$-axis. The $y=x$ line (dashed blue) is for reference.}
\centering
\includegraphics[width=0.49\textwidth]{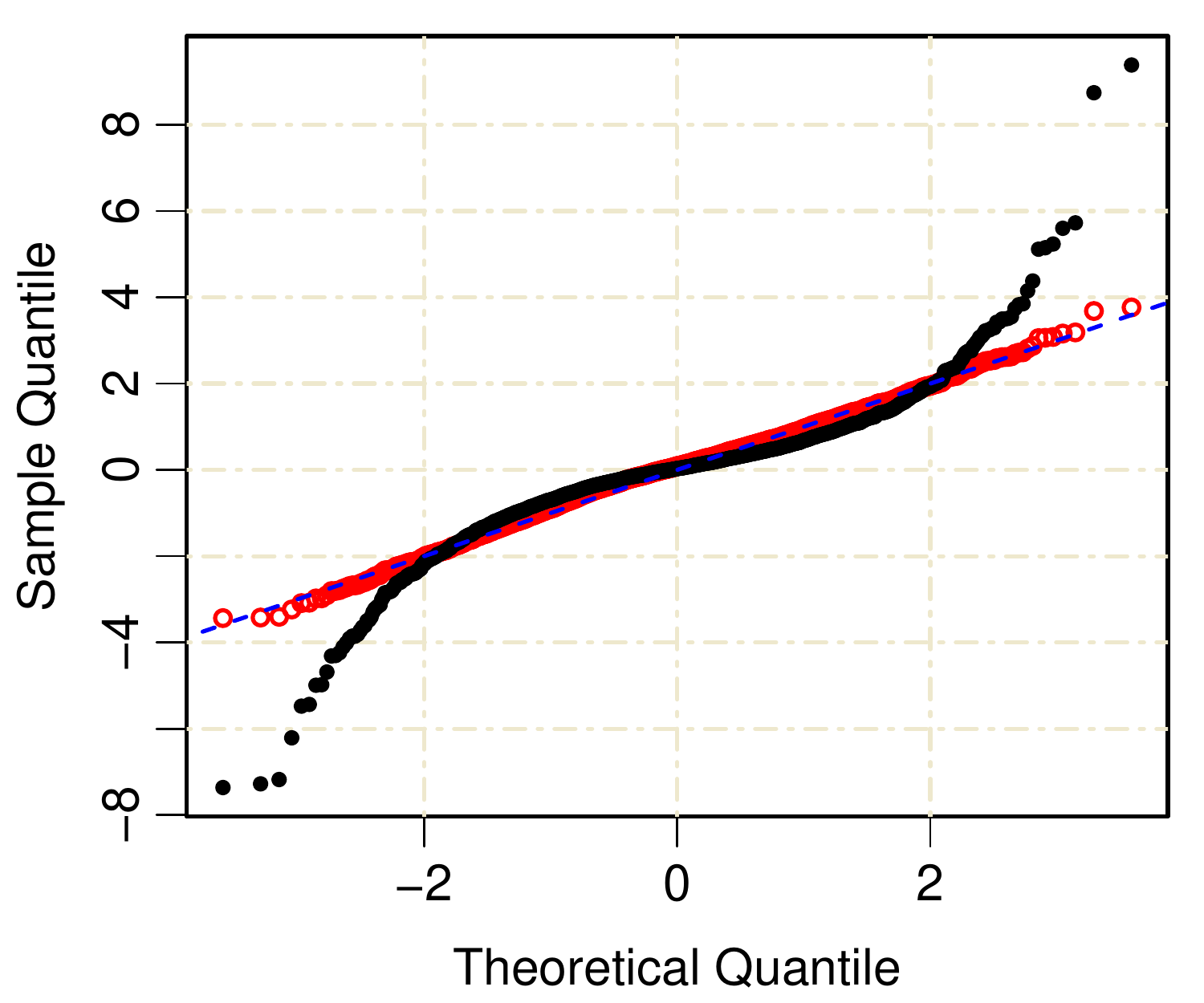} 
\includegraphics[width=0.49\textwidth]{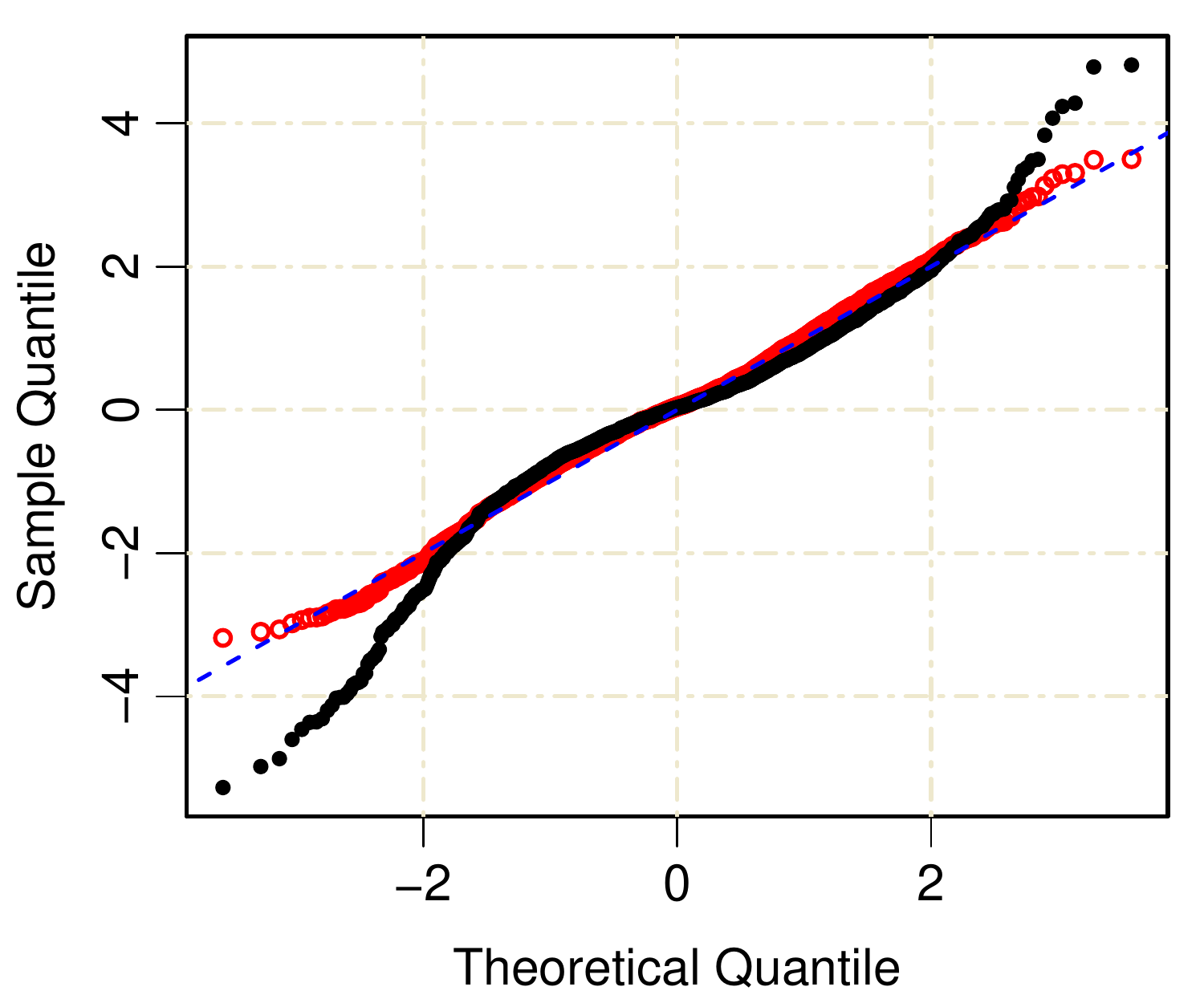}
\end{figure}

\section{Conclusion} \label{sec:conc}
The study generalizes the arithmetic Brownian motion with stochastic volatility. The NSVh process proposed in this study incorporates the SABR model and Johnson's $S_U$ distribution, which have been studied in different contexts. The SABR model is a well-known option pricing model in financial engineering, and the Johnson's $S_U$ distribution is a popular skewed and heavy-tailed distribution defined by a transformation from the normal random variable. From the generalizations of Bougerol's identity, the NSVh model is equipped with closed-form MC simulation for general cases and closed-form option formula for the $S_U$ case. This study demonstrates the usage of the model with two empirical datasets---the US swaption and daily returns distribution of the S\&P 500 and CSI 300 indexes. 

\section*{Acknowledgements}
The authors are grateful to Larbi Alili for sharing manuscripts, \citet{alili1997} and \citet{alili1997iden}, Robert Webb (editor), and Minsuk Kwak (discussant at the 2018 Asia-Pacific Association of Derivatives conference in Busan, Korea). Jaehyuk Choi would like to express gratitude to his previous employer, Goldman Sachs, for laying the foundation for the motivation of this study. Byoung Ki Seo was supported by the Institute for Information \& communications Technology Promotion (IITP) grant funded by the Ministry of Science and ICT (MSIT), Korea (No. 2017-0-01779, A machine learning and statistical inference framework for explainable artificial intelligence).
\bibliographystyle{plainnat}
\bibliography{../../@Bib/SABR,../../@Bib/normalmodel}
\begin{appendix}

\section{Proof of Proposition \ref{prop:boug1}} \label{sec:append1}
\boughyp*

\begin{proof}
Let $(x_t,y_t,z_t)$ be a three-dimensional hyperbolic BM in $\Hplane{3}$ with a drift on $z$-axis, starting at $(x_0,y_0,z_0)=(0,0,1)$:
$$ dx_t = z_t \,dX_t,\quad dy_t = z_t \,dY_t, \quad \text{and}\quad
\frac{dz_t}{z_t} = dZ_t + \left(\frac12 + \mu \right) dt,
$$
where the drift $\mu$ will be replaced by $(\drift-1)/2$ in the NSVh model. The standard BM in $\Hplane{3}$ introduced in Table~\ref{tab:hyp_geo} corresponds to $\mu=\drift=-1$.
Evidently,
$$ x_T \distequal y_T \distequal X_{\ExpBM_T} \quad\text{and}\quad 
z_T = \exp\left(\Zdrift_T\right)$$.
If we let $D_t$ be the hyperbolic distance between $(x_t, y_t, z_t)$ and the starting point $(0,0,1)$,then$$ D_t = \acosh \left( \frac12 \left( \frac{x_t^2+y_t^2}{z_t} + z_t + \frac1{z_t} \right) \right),$$
the Euclidean radius of $(x_t,y_t)$ is expressed by the function $\phi$:
\begin{equation} \label{eq:phi_r}
r_t = \sqrt{x_t^2+y_t^2} = \phi\hspace{-0.25em}\left(\Zdrift_t,\, D_t\right).
\end{equation}
The underlying BM $\Zdrift_t$ can be also interpreted as the projection of $(x_t,y_t,z_t)$ on the $z$-axis, that is, the signed hyperbolic distance from $(0,0,1)$ to $(0, 0, z_t)$. Therefore, the restriction $\Zdrift_t \le D_t$ is naturally satisfied.

A critical step of the proof is to show, for a fixed time $T$, 
\begin{equation} \label{eq:dist_hyp_euc}
D_T \distequal \sqrt{X_T^2 + Y_T^2 + (\Zdrift_T)^2} \qtext{conditional on} \Zdrift_T,
\end{equation},
which effectively means that the hyperbolic distance between $(x_T,y_T,z_T)$ and the starting point $(0,0,1)$ has the same distribution as the Euclidean distance of the underlying BMs, $(X_T, Y_T, \Zdrift_T)$, from $(0,0,0)$. Furthermore, the identity holds, conditional on $\Zdrift_T$. Based on the identity, it follows that
\begin{gather*}
\sqrt{x_T^2 + y_T^2} \;\distequal\;
\phi\hspace{-0.25em} \left(\Zdrift_T,\;\sqrt{X_T^2 + Y_T^2+(\Zdrift_T)^2}\right), \\
\text{or}\quad x_T \distequal \cos\theta \;\phi\hspace{-0.25em} \left(\Zdrift_T,\; \sqrt{X_T^2 + Y_T^2+(\Zdrift_T)^2}\right),
\end{gather*}
where $\theta$ is a uniformly distributed random angle.

Proving Equation~(\ref{eq:dist_hyp_euc}) for just one value of $\mu$ is enough because the rest follows from the Girsanov's theorem. Deviating from the original proof, $\mu=-1$ is chosen to take advantage of the $\Hplane{3}$ heat kernel. From the derivative of Equation~(\ref{eq:phi_r}), $r_T dr_T = z_T \sinh D_T dD_T$, the joint PDF on $r_T$ and $z_T$
is obtained from $p_3(t,D)$:
\begin{equation*}
\text{Prob}(r_T\in dr_T, z_T\in dz_T) = p_3(T,D_T) \frac{2\pi r_T\; dr_T dz_T}{z_T^3} = 
\frac1{\sqrt{2\pi T^3}} D_T\; e^{-\frac{1}{2T}(T^2+D_T^2)} \frac{dD_T dz_T}{z_T^2}.
\end{equation*}
From $dz_T/z_T = d\Zdrift[-1]_T$, it is also seen that
$$\text{Prob}(z_T\in dz_T) = \frac{1}{\sqrt{2\pi T}}e^{-\frac{1}{2T} Z_T^2}\, \cdot \frac{dz_T}{z_T}.$$
Therefore, the conditional probability is given as
$$
\text{Prob}(r_T\in dr_T|Z_T) = \frac{\text{Prob}(r_T\in dr_T, z_T\in dz_T)}{\text{Prob}(z_T\in dz_T)} = \frac{D_T}{T}\; e^{-\frac{1}{2T}\left(D_T^2+T^2-Z_T^2\right)} \frac{dD_T}{z_T}
=\frac{D_T}{T}\; e^{-\frac{1}{2T}\left(D_T^2-(\Zdrift[-1]_T)^2\right)} dD_T.
$$
The probability can be interpreted as the conditional probability $\text{Prob}(\sqrt{X_T^2+Y_T^2+(\Zdrift[-1]_T)^2}\in d D_T \;|\; \Zdrift[0]_T)$.
\end{proof}

\section{Derivation of moments of NSVh distribution}\label{sec:append2} 
\nsvhmoments*
\begin{proof}
We first compute the moments, conditional on $\Zdrift_T$:
\begin{align*}
E\Big(X_{\ExpBM_T}^{2n} \Big| \Zdrift_T \Big) &= E(\cos^{2n}\theta) \; E\left(\phi^{2n}\!\left(\Zdrift_T,\; \sqrt{X_T^2 + Y_T^2+(\Zdrift_T)^2}\right)\right) \\
&= \frac{(2n-1)!!}{n!} T^n\, e^{nu\sqrt{T}} \int_{u}^\infty r e^{-\frac12 r^2} \left( \cosh(r\sqrt{T})  - \cosh(u\sqrt{T}) \right)^n\; dr,
\end{align*}
where $u = |\Zdrift_T|/\sqrt{T}$ and $(2n-1)!! = (2n-1)(2n-3)\cdots3\cdot 1$.
This formula is comparable to the formula for $E\Big(\big({\ExpBM_T}\big)^n \Big| \Zdrift_T \Big)$ given in Proposition 5.3 in \citet{matsuyor2005exp1}. The first two values are computed in closed form:
\begin{align*}
E\Big( X_{\ExpBM_T}^2 \Big| \Zdrift_T \Big) &= T e^{u\sqrt{T}} m(u,\sqrt{T})  \\
E\Big( X_{\ExpBM_T}^4 \Big| \Zdrift_T \Big) &= 3 T^2 e^{2u\sqrt{T}} \left( m(u,2\sqrt{T}) - \cosh(u\sqrt{T})\,m(u,\sqrt{T})\right) \\
\text{where}\quad & m(u,\epsilon) = \frac{N(u+\epsilon)-N(u-\epsilon)}{2\epsilon\;e^{-\frac12 \epsilon^2}\; n(u)}.
\end{align*}
From these, the first two conditional moments of $\ExpBM_T$ are trivially obtained as
$$ E\Big(\ExpBM_T \Big| \Zdrift_T \Big) = E\Big(X_{\ExpBM_T}^2 \Big| \Zdrift_T\Big)
\qtext{and}
E\Big(\big(\ExpBM_T\big)^2 \Big| \Zdrift_T \Big) = \frac13 E\Big( X_{\ExpBM_T}^{\;4} \Big| \Zdrift_T\Big).$$ 
The same results for $\drift=0$ are derived in \citet{kennedy2012prob} in the context of the normal SABR model.

The unconditional moments of $X_{\ExpBM_S}$ are given as
\begin{align*}
E\Big(X_{\ExpBM_S}^{\;\;2}\Big) &= \frac{\ExpS^{2+2\mu} - 1}{2+2\mu}\\
E\Big(X_{\ExpBM_S}^{\;\;4}\Big) &= 
\frac{3}{2} \left( -\ExpS^{2+2\mu}\;\frac{\ExpS^{6+2\mu}-1}{6+2\mu} + (\ExpS^{4+2\mu}+1)\frac{\ExpS^{4+2\mu}-1}{4+2\mu} - \frac{\ExpS^{2+2\mu}-1}{2+2\mu} \right),
\end{align*}
where $\ExpS = e^{S}$. Subsequently, the expressions for the moments follow from
\begin{align*}
E\Big(\nondim{F}_S^{2}\Big) &= E\Big( \rho^2( e^{\Zdrift[\mu]_S} -  e^{\drift S/2} )^2 + \rhoc^2 E\Big(X_{\ExpBM_S}^{\;\;2}\Big| \Zdrift_T \Big)\; \Big) \\
E\Big(\nondim{F}_S^{3}\Big) &= E\Big( \rho^3( e^{\Zdrift[\mu]_S} -  e^{\drift S/2} )^3 + 3 \rho \rhoc^2( e^{\Zdrift[\mu]_S} -  e^{\drift S/2} )\,E\Big(X_{\ExpBM_S}^{\;\;2}\Big| \Zdrift_T \Big)\; \Big) \\
E\Big(\nondim{F}_S^{4}\Big) &= E\Big( \rho^4( e^{\Zdrift[\mu]_S} -  e^{\drift S/2} )^4 + 6 \rho^2 \rhoc^2( e^{\Zdrift[\mu]_S} -  e^{\drift S/2} )^2\,E\Big(X_{\ExpBM_S}^{\;\;2}\Big| \Zdrift_T \Big) + \rhoc^4 E\Big(X_{\ExpBM_S}^{\;\;4}\Big| \Zdrift_T \Big)\; \Big)
\end{align*}
and the substitution $\mu = (\drift - 1)/2$.
\end{proof}

\end{appendix}
\end{document}